\documentclass{llncs}

\usepackage[lncs,amsthm,autoqed]{pamas}
\usepackage{arydshln}
\usepackage{wrapfig}
\usepackage{amsmath}
\usepackage{paralist}

\usepackage{enumitem,lipsum}
\usepackage{bm}

\usepackage{algorithm}
\usepackage{algorithmic}
\usepackage{nicefrac}
\usepackage{multirow}

\usepackage{tikz}
\usetikzlibrary{positioning,chains,fit,shapes,calc}
\usepackage{color,soul}

\usetikzlibrary{arrows}

\usepackage{enumitem,lipsum}

\newcommand\eat[1]{}
\usepackage{tikz}
\usetikzlibrary{arrows}

\usepackage{mathrsfs}
\usepackage{slashbox}

\usepackage{thmtools}
\usepackage{thm-restate}

\newcommand{\pref}{{\succ}} % double brackets mean it isn't treated as a relation

\newcommand{\man}{m}
\newcommand{\woman}{w}

\newcommand{\prob}{p}
\newcommand{\matching}{\mu}

\newcommand{\certain}{\text{cert}}

\newcommand{\mtext}[1]{\text{\normalfont #1}}

\newcommand{\midd}{\makebox[2ex]{$:$}}

\definecolor{kentuckyblue}{RGB}{0, 93, 170}			%Go Big Blue!
\definecolor{green}{RGB}{0, 102, 0}					%Haris Green
\definecolor{frenchred}{RGB}{250,60,50}				% French Red
\definecolor{orange}{RGB}{255,127,0}             % orange

\title{Stable Matching\\ with Uncertain Linear Preferences \protect\footnote{A preliminary version of this paper has been accepted for publication in the proceedings of the 9th International Symposium on Algorithmic Game Theory (SAGT 2016).}}

\author{%
Haris Aziz \inst{1} \inst{2}
\and P{\'e}ter Bir{\'o} \inst{3}
\and Serge Gaspers \inst{2} \inst{1}
\and Ronald de Haan \inst{4}
\and Nicholas Mattei \inst{1} \inst{2}
\and Baharak Rastegari \inst{5}
}
\institute{ %
	Data61 (formerly: NICTA), CSIRO, Sydney, Australia.
	\email{haris.aziz@data61.csiro.au}, \email{nicholas.mattei@data61.csiro.au}
	\and
	University of New South Wales, Sydney, Australia.
	\email{sergeg@cse.unsw.edu.au}
	\and
	Hungarian Academy of Sciences, Institute of Economics. \email{peter.biro@krtk.mta.hu}
	\and
	TU Wien, Vienna, Austria.
	\email{dehaan@ac.tuwien.ac.at}
	\and
	School of Computing Science, University of Glasgow, Glasgow, UK.
	\email{baharak.rastegari@glasgow.ac.uk}
}

\begin{document}

\maketitle

\begin{abstract}
    We consider the two-sided stable matching setting in which there may be uncertainty about the agents' preferences due to limited information or communication. We consider three models of uncertainty: (1) lottery model --- in which for each agent, there is a probability distribution over linear preferences, (2) compact indifference model --- for each agent, a weak preference order is specified and each linear order compatible with the weak order is equally likely and (3) joint probability model --- there is a lottery over preference profiles.
For each of the models, we study the computational complexity of computing the stability probability of a given matching as well as finding a matching with the highest probability of being stable. We also examine more restricted problems such as deciding whether a certainly stable matching exists.
We find a rich complexity landscape for these problems, indicating that the form uncertainty takes is significant.
\end{abstract}

\section{Introduction}
We consider a \emph{Stable Marriage problem (SM)}
%\serge{Replace all occurrences of ``marriage'' with matching?}
%peter2: Stable Marriage problem is the classical name for one-to-one stable matching problems, so I would just stick with that...
 in which there is a set of men and a set of women.
%\haris{do we stick to equal number of men and women? we don't need it anywhere. }
%bahar2: don't we even need the equal number for men and women in Theorem 2?
Each man has a linear order over the women, and each woman has a linear order over the men. For the purpose of this paper we assume that the preference lists are complete, i.e., each agent finds each member of the opposite side acceptable.\footnote{We note that the complexity of all problems that we study are the same for complete and incomplete lists, where non-listed agents are deemed unacceptable---see Proposition \ref{complete} in Section \ref{sec:complete}. 
%which explains that our efficient algorithms described for the case of complete lists can be extended to incomplete lists. Additionally, our hardness proofs for incomplete lists can be transformed for complete lists. In fact, all our hardness reductions, except Theorem~\ref{MatchingWithHighestStabilityProbability-compact}, are for complete lists so they trivially extend to the case of incomplete lists.
}
%peter2: as I noted above I decided to follow Haris' suggestion, and so I added the comment below:
%bahar2: I'm rephrasing the text below slightly.
In the stable marriage problem the goal is to compute a \emph{stable matching}; a matching where no two agents prefer to be matched to each other rather than be matched to their current partners.
Unlike most of the literature on stable matching problems~\citep{GuIr89a,Manl13a,RoSo90a}, we assume that men and women may have uncertainty in their preferences which can be captured by various probabilistic uncertainty models. We focus on \emph{linear models} in which each possible deterministic preference
profile is a set of linear orders.

Uncertainty in preferences could arise for a number of reasons both practical or epistemological. For example, an agent could express a weak order because the agent did not invest enough time or effort to differentiate between potential matches and therefore one could assume that each linear extension of the weak order is equally likely; this maps to our \emph{compact indifference model}.
In many real applications the ties are broken randomly with lotteries, e.g., in the school choice programs in New York and Boston as well as in centralized college admissions in Ireland. However, a central planner may also choose a matching that is optimal in some sense, without breaking the ties in the preference list. For instance, in Scotland they used to compute the maximum size (weakly) stable matching to allocate residents to hospitals \cite{Manl13a}. We argue that another natural solution could be the matching which has the highest probability of being stable after conducting a lottery.
Alternatively, there may be a cost associated with eliciting preferences from the agents, so a central planner may want to only obtain and provide a recommendation based on a subset of the complete orders \cite{DrBo14a}.

As another example, imagine a group of interns are admitted to a company and allocated to different projects based on their preferences and the preferences of the project leaders. Suppose that after three months the interns can switch projects if the project leaders agree; though the company would prefer not to have swaps if possible. However, both the interns and the project leaders can have better information about each other after the three months, and the assignment should also be stable with regard to the refined preferences. This example motivates our lottery and joint probability models. In the \emph{lottery model}, the agents have \emph{independent} probabilities over possible linear orders (e.g.\ each project leader has a probability distribution on possible refined rankings over the interns independently from each other). In the \emph{joint probability model}, the probability distribution is over possible preference profiles and can thus accommodate the possibility that the preferences of the agents are refined in a correlated way (e.g.\ if an intern performs well in the first three months then she is likely to be highly ranked by all project leaders). Uncertainty in preferences has already been studied in voting~\citep{HAK+12a} and for cooperative games \cite{LC15}. Ehlers and Mass\'o \cite{EhlersMasso2015} considers many-to-one matching markets under a Bayesian setting. Similarly, in auction theory, it is standard to examine Bayesian settings in which there is a probability distribution over the types of agents.

To illustrate the problem we describe a simple example with four agents. We write $b~\pref_a c$ to say that agent $a$ prefers $b$ to $c$
and assume the lottery model.
%peter2: I added back the above line, as I found this important to stress

\begin{example}\label{ex:1}
We have two men $\man_1$ and $\man_2$ and two women $\woman_1$ and $\woman_2$. Each agent assigns a probability to each strict preference ordering as follows.
\begin{inparaenum}[(i)]
	\item $\prob(\woman_1 \pref_{\man_1} \woman_2 )=0.4$ and $\prob(\woman_2 \pref_{\man_1} \woman_1 )=0.6$
	\item $\prob(\woman_1 \pref_{\man_2} \woman_2 )=0.0$ and $\prob(\woman_2 \pref_{\man_2} \woman_1 )=1.0$
	\item $\prob(\man_1 \pref_{\woman_1} \man_2 )=1.0$ and $\prob(\man_2 \pref_{\woman_1} \man_1 )=0.0$
	\item $\prob(\man_1 \pref_{\woman_2} \man_2 )=0.8$ and $\prob(\man_2 \pref_{\woman_2} \man_1 )=0.2$.
\end{inparaenum}
%As every agent finds every candidate acceptable,
This setting admits two matchings that are stable with positive probability: $\matching_1 = \{(\man_1,\woman_1), (\man_2,\woman_2)\}$ and $\matching_2 = \{(\man_1,\woman_2), (\man_2,\woman_1)\}$. Notice that if each agent submits the preference list that s/he finds most likely to be true, then the setting admits a unique stable matching that is $\matching_2$. The probability of $\matching_2$ being stable, however, is $0.48$ whereas the probability of $\matching_1$ being stable is $0.52$.
\end{example}

\subsection{Uncertainty Models}

We consider three different uncertainty models:
%of \emph{pairwise probability}, \emph{lottery}, and \emph{compact indifference}
%defined as follows.

\begin{itemize}

\item \textbf{Lottery Model}: For each agent, we are given a probability distribution over strict preference lists.

\item \textbf{Compact Indifference Model}: Each agent reports a single weak preference list that allows for ties. Each
complete linear order extension of this weak order is assumed to be equally likely.
\item  \textbf{Joint Probability Model}: A probability distribution over preference profiles is specified.
\end{itemize}

Note that for the Lottery Model and the Joint Probability Model the representation of the input preferences
can be exponentially large.  However, in settings where similar models of uncertainty are used, including resident matching \cite{DrBo14a} and voting \cite{HAK+12a}, a
%it is common to expect and see in real world data limited amounts of uncertainty, i.e., small supports.
limited amount of uncertainty (i.e. small supports) is commonly expected and observed in real world data.
%Typically, agents can express parts of their preference lists, the things they either really prefer or the things they do not want,
%and are only uncertain in very limited ways.  
Consequently, we consider
%the special case presented
special cases
when the uncertainty is bounded in certain natural ways including the existence of only a small number of uncertain preferences and/or uncertainty on only one side of the market.

Observe that the compact indifference model can be represented as a lottery model.
This is a special case of the lottery model in which each agent expresses a weak order over the candidates (similar to the SMT setting \cite{GuIr89a,Manl13a}).
However, the lottery model representation can be exponentially larger than the compact indifference model; for an agent that is indifferent among $n$ agents on the other side of the market, there are $n!$ possible linearly ordered preferences.

\subsection{Computational Problems}

Given a stable marriage setting where agents have uncertain preferences, various natural computational problems arise.
Let \emph{stability probability} denote the probability that a matching is stable. We then consider the following two natural problems for each of our uncertainty models.

\begin{itemize}
    \item {\sc MatchingWithHighestStabilityProbability}:  Given uncertain preferences of the agents, compute a matching with the highest stability probability.
    \item {\sc StabilityProbability}: Given a matching and uncertain preferences of the agents, what is the stability probability of the matching?
    % \item  {\sc CompareStability}: given two matchings, which matching has a higher probability of being stable?
    %\item {\sc MatchingWithMinBlockingPairs}: Compute a matching with minimum expected number of blocking pairs.
\end{itemize}

%{\sc MatchingWithMinBlockingPairs}: Compute a matching with minimum expected number of blocking pairs.

\noindent
We also consider two specific problems that are simpler than {\sc StabilityProbability}: (1) {\sc IsStabilityProbabilityNon-Zero} --- For a given matching, is its stability probability non-zero? (2) {\sc IsStabilityProbabilityOne} --- For a given matching, is its stability probability one?

% \begin{itemize}
%     \item {\sc IsStabilityProbabilityNon-Zero}: Given uncertain preferences of the agents and a matching, is the probability of the matching being stable non-zero?
%     \item {\sc IsStabilityProbabilityOne}: For a given matching, is the probability of being stable one?
% \end{itemize}
We additionally consider problems connected to, and more restricted than, {\sc MatchingWithHighestStabilityProbability}:
(1) {\sc ExistsCertainlyStableMatching} --- Does there exist a matching that has stability probability one?  (2) {\sc ExistsPossiblyStableMatching} --- Does there exist a matching that has non-zero stability probability?

%
% \begin{itemize}
%     \item {\sc ExistsCertainlyStableMatching}: Does there exist a matching that has probability of stability one?
%         \item {\sc ExistsPossiblyStableMatching}: Does there exist a matching that has probability of stability non-zero?
%     %    \item {\sc ExistsCertainlyStableMatching}: Does there exist a matching that has probability of stability one?
% \end{itemize}

Note that {\sc ExistsPossiblyStableMatching} is straightforward to answer for any of the three uncertainty models we consider here, since there exists a stable matching for each deterministic preference profile that is a possible realization of the uncertain preferences.

\begin{table}[tb]
    %    \small
    \centering
    \scalebox{0.9}{
    \begin{tabular}{lccc}
        \toprule
     &Lottery &Compact&Joint  \\
         Problems&Model&Indifference&Probability\\
         \midrule
\multirow{ 2}{*}{\sc StabilityProbability}&\#P-complete&?&in P\\
&\multicolumn{3}{c}{in P for all three models if 1 side is certain}\\
\midrule
{\sc IsStabilityProbabilityNon-Zero}&NP-complete&in P&in P\\
{\sc IsStabilityProbabilityOne}&in P&in P&in P\\
\midrule
{\sc ExistsPossiblyStableMatching}&in P&in P&in P\\
{\sc ExistsCertainlyStableMatching}&in P &in P&NP-complete\\
\midrule
 \multirow{ 3}{*}{\sc MatchingWithHighestStabilityProb}&?&NP-hard&NP-hard\\
 &\multicolumn{3}{c}{in P for all models if 1 side is certain and}\\ % $O(1)$ uncertain agents}\\
 &\multicolumn{3}{c}{there is $O(1)$ number of uncertain agents}\\
%{\sc MatchingWithMinBlockingPairs}&?&?&?\\
        \bottomrule
    \end{tabular}
    }
    \caption{\label{table:summary:uncertainmatch}Summary of results.
    }

\end{table}

\vspace{-1em}

\subsection{Results}
Table \ref{table:summary:uncertainmatch} summarizes our main findings. Note that the complexity of each problem is considered with respect to the input size, and that under the lottery and joint probability models the size of the input could be exponential in $n$, namely $O(n!\cdot 2n)$ for the lottery model and $O((n!)^{2n})$ for the joint probability model, where $n$ is the number of agents on either side of the market.
%The complete version of the sketched or missing proofs can be found in the full version of the paper \cite{}.

We point out that {\sc StabilityProbability} is \#P-complete for the lottery model even when each agent has at most two possible preferences, but in P if one side has certain preferences.
% Not very important, and trivial to deduce from other results:
%Furthermore, unless $\text{NP}=\text{RP}$, there is no FPRAS for {\sc StabilityProbability} under the lottery model.
Additionally, we show that {\sc IsStabilityProbabilityNon-Zero} is in P for the lottery model if each agent has at most two possible preferences. Note that {\sc StabilityProbability} is open for the compact indifference model when both sides may be uncertain, and we also do not know the complexity of {\sc MatchingWithHighestStabilityProbility} in the lottery model, except when only a constant number of agents are uncertain on the same side of the market.

\section{Preliminaries}

In the Stable Marriage problem, there are two sets of agents. Let $M$ denote a set of $n$ men and $W$ a set of $n$ women. We use the term \emph{agents} when making statements that apply to both men and women, and the term \emph{candidates} to refer to the agents on the opposite side of the market to that of an agent under consideration. Each agent has a linearly ordered preference over the candidates. An agent may be uncertain about his/her linear preference ordering. %Let $L_a$ denote the probabilistic preferences for agent $a$ under a given uncertainty model.
Let $L$ denote the \emph{uncertain preference profile} for all agents. We denote by $I=(M,W,L)$ an instance of a \emph{Stable Marriage problem with Uncertain Linear Preferences (SMULP)}.

We say that a given uncertainty model is \emph{independent} if any uncertain preference profile $L$ under the model can be written as a product of uncertain preferences $L_a$ for all agents $a$, where all $L_a$'s are independent. Note that the lottery and the compact indifference models
are both independent, but the joint probability model is not.

A \emph{matching} $\mu$ is a pairing of men and women such that each man is paired with at most one woman and vice versa; defining a list of (man, woman) pairs $(m,w)$.  We use $\mu(m)$ to denote the woman $w$ that is matched to $m$ and $\mu(w)$ to denote the match for $w$. Given linearly ordered preferences, a matching is \emph{stable} if there is no pair $(m,w)$ not in $\mu$ where $m$ prefers $w$ to his current partner in $\mu$, i.e., $w \succ_m \mu(m)$, and vice versa. If such a pair exists, it constitutes a \emph{blocking pair}; as the pair would prefer to defect and match with each other rather than stay with their partner in $\mu$. Given an instance of SMULP, a matching is \emph{certainly stable} if it is stable with probability 1.

The following extensions of SM will come in handy in proving our results. The \emph{Stable Marriage problem with Partially ordered lists (SMP)} is an extension of SM in which agents' preferences are partial orders over the candidates.
The \emph{Stable Marriage problem with Ties (SMT)} is a special case of SMP in which incomparability is transitive and is interpreted as indifference. Therefore, in SMT each agent partitions the candidates into different ties (equivalence classes), is indifferent between the candidates in the same tie, and has strict preference ordering over the ties. In some practical settings some agents may find some candidates unacceptable and prefer to remain unmatched than to get matched to the unacceptable ones. \emph{SMP with Incomplete lists (SMPI)} and \emph{SMT with Incomplete lists (SMTI)} captures these scenarios where each agent's partially ordered list contains only his/her acceptable candidates. A matching is \emph{super-stable} in an instance of SMPI if it is stable w.r.t. all linear extensions of the partially ordered lists.

We define the \emph{certainly preferred} relation $\succ_a^{\certain}$ for agent $a$. We write $b \succ_a^{\certain} c$ if and only if agent $a$ prefers $b$ over $c$ with probability 1.  Based on the certainly preferred relation, we can define a dominance relation $D$:
$D_{m}(w)=\{w\}\cup \{w'\midd w'\succ_m^{\certain} w\}$;
$D_{w}(m)=\{m\}\cup \{m'\midd m'\succ_w^{\certain} m\}$.
Based on the notion of the dominance relation, we present a useful characterization of certainly stable matchings for independent uncertainty models.

\begin{restatable}{lemma} %[Characterization of Certainly Stable Matchings]
{lemCharact}\label{lemma:charac-certainly-stable}
A matching $\mu$ is certainly stable for an independent uncertainty model if and only if for each pair $\{m,w\}$, $\mu(m)\in D_m(w)$ or $\mu(w)\in D_w(m)$.
\end{restatable}
\begin{proof} %[of Lemma ~\ref{lemma:charac-certainly-stable}]
Assume that there exists a pair $\{m,w\}$ such that $\mu(m)\notin D_m(w)$ or $\mu(w)\notin D_w(m)$. Then, $m$ has non-zero probability of preferring $w$ over $\mu(m)$ and $w$ has non-zero probability of preferring $m$ over $\mu(w)$. But this means that $\mu$ has non-zero probability of not being stable.

Assume that a matching $\mu$ is certainly stable. Then no blocking pair $\{m,w\}$ has non-zero probability of forming.  This is only possible if the pair $\{m,w\}$ is part of the matching or one of $m$ and $w$ have zero probability of preferring the blocking $\{m,w\}$ over their current match in $\mu$.
%\nick{slightly modified the previous sentence for clarity}
%\haris{thanks}
 In either case, $\mu(m)\in D_m(w)$ or $\mu(w)\in D_w(m)$.
\end{proof}

\noindent
We point out that certainly preferred relation can be computed in polynomial time for all three models studied in this paper.

Certainly stable matchings are closely related to the notion of super-stable matchings %widely studied in the literature
\cite{GI89,Irv94}. In fact we can define a certainly stable matching using a terminology similar to that of super-stability. Given a matching $\mu$ and an unmatched pair $\{m,w\}$, we say that $\{m,w\}$ \emph{very weakly blocks (blocks)} $\mu$ if $\mu(m) \not\succ_m^{\certain} w$ and $\mu(w) \not\succ_w^{\certain} m$. The next claim then follows from Lemma \ref{lemma:charac-certainly-stable}.
%$m$ does not certainly prefer $\mu(m)$ to $w$ and $w$ does not certainly prefer $\mu(w)$ to $m$.

\begin{proposition}\label{prop:no-blocking-pair}
A matching $\mu$ is certainly stable for an independent uncertainty model if and only if it admits no very weakly blocking pair.
\end{proposition}

%%%%%%%%%%%%%%%%%
%
%%%%%%%%%%%%%%%%%
%\section{Existence of a Certainly Stable Matching}
\section{General Results}
\label{sec:certain}

In this section, we first show that the complexity of all problems that we study are the same for complete and incomplete lists. We then present some general results that apply to multiple uncertainty models. 
%First
We show that {\sc ExistsCertainlyStableMatching} can be solved in polynomial time for any independent uncertainty model including lottery and compact indifference.
We then prove that, when the number of uncertain agents is constant and one side of the market is certain, we can solve {\sc MatchingWithHighestStabilityProbability} efficiently for each of the linear models.

%%%%%%%%%%%%%%%%%%%%%%%%%%%%%%%%%%%%%%%%%%%%
\subsection{The Case for Incomplete Lists}\label{sec:complete}
The next proposition explains that our efficient algorithms described for the case of complete lists can be extended to incomplete lists. Additionally, our hardness proofs for incomplete lists can be transformed for complete lists. In fact, all
our hardness reductions, except Theorem~\ref{MatchingWithHighestStabilityProbability-compact}, are for complete lists so they trivially extend to the case of incomplete lists.

\begin{proposition}\label{complete}
The complexity of each computational problem studied in this paper are the same for complete and incomplete lists. Formally, if $I$ is a linear model with incomplete lists then we can construct an instance $I'$ with complete lists such that for each matching $\mu$ in $I$ there exists a corresponding matching $\mu'$ in $I'$ with $p(\mu,I)=p(\mu',I')$, such that $\mu$ can be obtained from $\mu'$ in polynomial time. Furthermore, $\mu$ is one of the most stable matchings in $I$ if and only if the corresponding matching $\mu'$ is one of the most stable matching in $I'$. Therefore a polynomial time algorithm solving {\sc StabilityProbability} or {\sc MatchingWithHighestStabilityProbability} for complete lists can be used to solve the same problem for incomplete lists in polynomial time.
\end{proposition}
\begin{proof}
In the case of complete lists we assumed that we have an equal number of men and women and everybody finds all candidates acceptable. When we consider the problem with incomplete lists we mean that the sizes of the two sets are not necessarily the same and not all the candidates are acceptable for the agents. However, we assume that in all realization of the preference profiles the same candidates are acceptable, so we only randomize on the preferences over the acceptable partners. Suppose that $I$ is an instance of a probabilistic model with incomplete lists with sets $M$ and $W$. Let us create the corresponding instance $I'$ with sets $M'$ and $W'$ in the following way. First we ensure that $|M'|=|W'|$ by adding enough agents to the short side of the market. Then we complete the preference lists of each agent by adding the previously unacceptable candidates to the end of her/his list according to a predetermined order, e.g. by the indices of the agents. Suppose now that $\mu$ is a matching in $I$ and $X$ is the set of matched men in $M$, whilst $\mu(X)=Y$. Let us create a corresponding matching $\mu'$ in $I'$ by extending $\mu$ with the unique stable matching for the subinstance restricted to the unmatched agents. Namely, let $\mu_u$ be the stable matching that matches $M'\setminus X$ to $W'\setminus Y$ in such a way that the $k$th pair contains the $k$th man and the $k$th woman from $M'\setminus X$ and $W'\setminus Y$, respectively according to their indices, and let $\mu'=\mu\cup\mu_u$. Now we claim that $p(\mu,I)=p(\mu',I')$. This is because there is no blocking pair in $(M'\setminus X)\times (W'\setminus Y)$, and any other pair is blocking for some preference profile in $I$ if and only if it is blocking for the corresponding preference profile in $I'$, obviously. Furthermore, it is also clear that among the extensions of $\mu$, $\mu'$ is the most stable one in $I'$. Therefore $\mu$ is one of the most stable matchings in $I$ if and only if the corresponding extension, $\mu'$ is one of the most stable matchings in $I'$. Thus an efficient algorithm for {\sc StabilityProbability} or {\sc MatchingWithHighestStabilityProbability} (or other subproblems) for complete lists can also be used to solve the same problems for the case of incomplete lists. This also implies that any hardness result proved for incomplete lists holds also for complete lists.
\end{proof}

%%%%%%%%%%%%%%%%%%%%%%%%%%%%%%%%%%%%%%%%%%%%%%%%%%%%%%%%%%%%%%%%%%%%%%%%%
\subsection{An Algorithm for the Lottery and Compact Indifference Models}
\label{sec:certain-lottery-compact}

\begin{restatable}{theorem}{thmTrans}\label{th:exists-matchwithprob1-p-bothsidestransitive-baharpeter}
%For the lottery model, {\sc ExistsCertainlyStableMatching} can be solved in polynomial time.
For any independent uncertainty model in which the certainly preferred relation is transitive and can be computed in polynomial time, {\sc ExistsCertainlyStableMatching} can be solved in polynomial time. %\haris{what's the running time?}
\end{restatable}
\begin{proof} %[of Theorem~\ref{th:exists-matchwithprob1-p-bothsidestransitive-baharpeter}]
We prove this by reducing {\sc ExistsCertainlyStableMatching} to the problem of deciding whether an instance of SMP admits a super-stable matching or not. The latter problem can be solved in polynomial time using algorithm SUPER-SMP in \cite{RCIIL-ec14}.

Let $I=(M,W,L)$ be an instance of {\sc ExistsCertainlyStableMatching} under an independent uncertainty model, assuming that the certainly preferred relation is transitive and can be computed in polynomial time.
%{\sc ExistsCertainlyStableMatching} with the set men $M$, women $W$, and $L$
We construct an instance $I'=(M,W,p)$ of SMP, in polynomial time, as follows. The set of men and women are unchanged. To create the partial preference ordering $p_a$ for each agent $a$ we do the following. W.l.o.g. assume that $a$ is a man $m$. For every pair of women $w_1$ and $w_2$ (i) if $w_1 \succ_m^{certain} w_2$ then $(w_1,w_2)\in p_m$, denoting that $m$ (strictly) prefers $w_1$ to $w_2$ in $I'$, (ii) if $w_2 \succ_m^{certain} w_1$ then $(w_2,w_1)\in p_m$, denoting that $m$ (strictly) prefers $w_2$ to $w_1$ in $I'$. We claim, and show, that $I'$ admits a super-stable matching if and only if $I$ admits a matching with stability probability one.
% \haris{I suggesting to send the remaining proof to the appendix.}
  A matching $\mu$ is super-stable in $I'$ if and only if it does not admit a very weak blocking pair. A pair $(m,w)$ unmatched in $\mu$ is a very weak blocking pair if (i) $m$ either prefers $w$ to $\mu(m)$ or is indifferent between them, and (ii) $w$ either prefers $m$ to $\mu(w)$ or is indifferent between them. Agent $a$ is indifferent between agents $b$ and $c$ under an SMP instance if neither $(b,c)$ nor $(c,b)$ is in $p_a$. It is easy to verify that an unmatched pair $(m,w)$ in $I'$ is a very weak blocking pair in $\mu$ if and only if $(\mu(m),w) \notin p_m$ and $(\mu(w),m) \notin p_w$.

% \nick{is there a reason the below is in emph?} \bahar{I have italicized the first sentence, that is the statement that is to be proved. The reason being to separate the claim from the rest of the paragraph that is the proof.}

\emph{Only if part: If $I'$ admits a super-stable matching $\mu$ then $\mu$ is certainly stable in $I$.} Assume for a contradiction that $\mu$ is not certainly stable in $I$. It then follows Lemma \ref{lemma:charac-certainly-stable} that $\mu(m)\notin D_m(w)$ and $\mu(w)\notin D_w(m)$, implying that $\mu(m)\not\succ_m^{strict} w$ and $\mu(w)\not\succ_w^{strict} m$, and thus $(\mu(m),w) \notin p_m$ and $(\mu(w),m) \notin p_w$. Therefore $(m,w)$ blocks $\mu$ in $I'$, a contradiction.

\emph{If part: If $I$ admits a certainly stable matching $\mu$ then $\mu$ is super-stable in $I'$.} Assume, for a contradiction, that $\mu$ is not super-stable in $I'$. Therefore there exists a very weak blocking pair $(m,w)$, implying that $(\mu(m),w) \notin p_m$ and $(\mu(w),m) \notin p_w$, which in turn implies that $\mu(m)\not\succ_m^{strict} w$ and $\mu(w)\not\succ_w^{strict} m$. The latter statement, coupled with the fact that $m$ and $w$ are not matched together, implies that $\mu(m)\notin D_m(w)$ and $\mu(w)\notin D_w(m)$, and thus by  Lemma \ref{lemma:charac-certainly-stable} $\mu$ is not stable in $I$, a contradiction.
\end{proof}
%\begin{proofsketch}
%We prove this by reducing {\sc ExistsCertainlyStableMatching} to the problem of deciding whether an instance of SMP admits a super-stable matching. The latter problem can be solved in polynomial time using algorithm SUPER-SMP in \cite{RCIIL-ec14}.
%
%Let $I=(M,W,L)$ be an instance of {\sc ExistsCertainlyStableMatching} under an independent uncertainty model, assuming that the certainly preferred relation is transitive and can be computed in polynomial time.
%%{\sc ExistsCertainlyStableMatching} with the set men $M$, women $W$, and $L$
%We construct an instance $I'=(M,W,p)$ of SMP, in polynomial time, as follows. The set of men and women are unchanged. To create the partial preference ordering $p_a$ for each agent $a$ we do the following. W.l.o.g., assume that $a$ is a man $m$. For every pair of women $w_1$ and $w_2$ (i) if $w_1 \succ_m^{\certain} w_2$ then $(w_1,w_2)\in p_m$, denoting that $m$ (strictly) prefers $w_1$ to $w_2$ in $I'$, (ii) if $w_2 \succ_m^{\certain} w_1$ then $(w_2,w_1)\in p_m$, denoting that $m$ (strictly) prefers $w_2$ to $w_1$ in $I'$. We claim, and show, that $I'$ admits a super-stable matching iff $I$ has a certainly stable matching.
%\end{proofsketch}

\subsection{An Algorithm for a Constant Number of Uncertain Agents}

\begin{restatable}{theorem}{thmConstantUncertain}\label{thm:highestprobability-constant}
When the number of uncertain agents is constant and one side of the market is certain then {\sc MatchingWithHighestStabilityProbability} is polynomial-time solvable for each of the linear models. %(Lottery, Joint probability, Compact Indifference).
%peter2: again, I commented out the list of problems
%\haris{Pls clarify efficient: FPT or polynomial given constant number of uncertain agents? what's the running time in terms of $k$? Is it proper FPT or simply polynomial given $k$ is constant?}
\end{restatable}
\begin{proof} %[of Theorem ~\ref{thm:highestprobability-constant}]
%\bahar{I prefer using $M$ and $W$ for the sets of men and women, and suggest using $\mu$ for matching. Though if you feel strongly about using $M$ for matching, I don't object. (How about using $U$ and $W$ for the sets of men and women? as it is in David's book). It's just that $A$ and $B$ sounds very random.}
%\haris{Let's keeps $\mu$s for matchings. Let's keep $M$ and $W$ for men and women. I agree with Bahar.  }
%peter: okay, I have replaced M with \mu everywhere.
Let $I=(M,W,L)$ be an instance of {\sc MatchingWithHighestStabilityProbability} and let $X\subseteq M$ be the set of uncertain agents with $|X|=k$ for a constant $k$. We consider all the possible matchings between $X$ and $W$, where their total number is $K=n(n-1)\dots (n-k)$. Let $\mu_i$ be such a matching for $i\in\{1\dots K\}$. The main idea of the proof is to show that there exist an extension of $\mu_i$ to $M\cup W$ that has stability probability at least as high as any other extension of $\mu_i$. In this way we will need to compute this probability for only a polynomial number of matchings in $n$, that we can do efficiently for each model, and then compare them and select the one with the highest probability.

So we take a matching $\mu_i$ between sets $X$ and $W$. Let $Y=\mu_i(X)$ (i.e., the partners of $X$ in $W$) and let $M'=M\setminus X$ and $W'=W\setminus Y$. First, we compute the man-optimal matching $\mu_i^M$ for the sub-instance $I'$ on $M'\cup W'$, that can be done efficiently by the Gale-Shapley algorithm \cite{GaSh62a}. Now, if there exist a blocking pair $\{m',w\}$ involving some certain agents $m\in M'$ and $w\in Y$ for $\mu\cup \mu_i^M$ in $I$, then we can conclude that any extension of $\mu_i$ for $I$ will have zero probability of being stable. This is because any extension of $\mu_i$ for $I$ that has a positive probability of being stable must also be stable for the sub-instance $I'$. If $\{m',w\}$ is a blocking pair for $\mu_i^M$ then it will remain blocking for any extension of $\mu_i$ for $I$ as well, since $w$ has the same partner and the $m'$ cannot have a better partner either. Thus we can exclude the extensions of $\mu_i$ from the further consideration in this case.

Suppose now that there is no blocking pair of the form $\{m',w\}$, as explained above, for $\mu_i^M$ in $I$.
%, for the matching $\mu_i$ when extended with $\mu_i^m$ in $I$.
We truncate the preference lists of men in $M'$ in the following way. For each man $m'\in M'$ we remove all the women $w'\in W'$ from the list of $m'$ that are less preferred by $m'$ than some woman in $Y$ that finds $m'$ better than her partner in $\mu_i$.
That is, we remove $w'$ from the list of $m'$ if there exists $w\in Y$ such that $m'\succ_{w} \mu_i(w)$ and $w\succ_{m'}w'$. Let us denote the sub-instance for $M'\cup W'$ with the truncated lists as $I_i^r$. Now we compute the woman-optimal matching, $\mu_i^W$ in $I_i^r$. Let $\mu_i^*=\mu_i\cup \mu_i^W$ be the extended matching in $I$. This is stable for the certain agents by the construction. 
%PB: I removed these sentences 
%To see this, first we note that the same agents are matched in both $\mu_i^W$ and $\mu_i^M$, since both are stable for $I_i^r$ (see e.g.\ \cite{Manl13a}). Suppose for a contradiction that $\mu_i^W$ is not stable for $I'$, then it must be blocked by a pair $(m',w')$, where $w'$ was truncated from the list of $m'$. But then $m'$ is matched in $\mu_i^W$ and has a better partner than $w'$, a contradiction.
%\bahar{I don't see why this is so obvious to not need an argument.}
%peter3: I added the explanation

Finally, we will show that for any matching $\mu_i'$, that is an extension of $\mu_i$ to $I$, the stability probability of $\mu_i'$ is less than, or equal to, the stability probability of $\mu_i^*$. If $\mu_i'$ is not stable for the certain agents then $\mu_i'$ has zero probability of being stable, thus the statement holds. If $\mu_i'$ is stable for the certain agents then it must also be stable in $I_i^r$, and each woman in $W'$ weakly prefers her partner in $\mu_i^*$ to her partner in $\mu_i'$, since she gets her optimal stable partner for $I_i^r$ in $\mu_i^*$. Therefore, if $\mu_i'$ is stable under a preference profile then $\mu_i^*$ will also be stable, so the statement follows.
 Thus, there remain only a polynomial number ($K$) of candidate matchings in $n$ for which we have to compute the probabilities. {\sc StabilityProbability}  is polynomial-time solvable for all the three models we consider given that one side has certain preferences, as described in Theorems \ref{thm:lottery-certain-stabilityprobability-poly}, \ref{StabilityProbability-compact}, and \ref{thm:joint-stabilityprobability-poly}.
\end{proof}
%\begin{proofsketch}
%Let $I=(M,W,L)$ be an instance of {\sc MatchingWithHighestStabilityProbability} and let $X\subseteq M$ be the set of uncertain agents with $|X|=k$ for a constant $k$. We consider all the possible matchings between $X$ and $W$, where their total number is $K=n(n-1)\dots (n-k)$. Let $\mu_i$ be such a matching for $i\in\{1\dots K\}$. The main idea of the proof is to show that there exist an extension of $\mu_i$ to $M\cup W$ that has stability probability at least as high as any other extension of $\mu_i$. In this way we will need to compute this probability for only a polynomial number of matchings in $n$, which we can do efficiently for each model when one side has certain preferences---see Theorems \ref{thm:lottery-certain-stabilityprobability-poly}, \ref{StabilityProbability-compact} and \ref{thm:joint-stabilityprobability-poly}, and select the one with the highest probability.
%\end{proofsketch}

\section{Lottery Model}

In this section we focus on the lottery model. %and the remaining unanswered questions.

\begin{restatable}{theorem}{thmCertstabLot}\label{thm:lottery-certain-stabilityprobability-poly}
For the lottery model, if one side has certain preferences,  {\sc StabilityProbability}  is polynomial-time solvable.
\end{restatable}
\begin{proof} %[of Theorem ~\ref{thm:lottery-certain-stabilityprobability-poly}]
Without loss of generality, assume that men have certain preferences.
The following procedure gives us the stability probability of $\mu$ for any given $\mu$. (1) For each uncertain woman $w$ identify those preferences that allow her not to form a blocking pair. We can do this in polynomial time as men have strict preferences and therefore for each preference ordering of $w$ we only need to look up the (one and only) preference ordering of each $m$ who $w$ prefers to $\mu(w)$. (2) For each uncertain woman $w$, add up the probabilities of all preference orderings that pass the test in the first step. (3) multiply the added-up probabilities for all $w$ obtained in step (2).
\end{proof}
%\begin{proofsketch}
%W.l.o.g. assume that men are certain. The stability probability of a given matching $\mu$ is equal to the probability that none of the possible blocking pairs form. The probability of one blocking pair $\{m,w\}$ forming is equal to the probability that $w$ prefers $m$ to $\mu(w)$ given $m$ also prefers $w$ to $\mu(m)$.
%%The stability probability of a given matching $\mu$ is equal to the probability that none of the possible blocking pairs form. The probability of one blocking pair forming is equal to the probability that the uncertain agent in the pair prefers the blocking partner given that the certain blocking partner also prefers the blocking pair.
%\end{proofsketch}

\begin{restatable}{theorem}{thmIsCertstabLot}\label{th:IsStabilityProbabilityOne-lottery}
For the lottery model, {\sc IsStabilityProbabilityOne} can be solved in linear time.
\end{restatable}
\begin{proof} %[of Theorem ~\ref{th:IsStabilityProbabilityOne-lottery}]
The problem is equivalent to checking whether the given matching $\mu$ has non-zero probability of \emph{not} being stable. This can be checked as follows. For each possible pair of agents $\{m,w\}$ that are not matched to each other, we check whether they can form a blocking pair with non-zero probability. For this, we just need to check whether $m$ prefers $w$ in some possible preference over $\mu(m)$ and whether $w$ prefers $m$ in some possible preference over $\mu(w)$.
\end{proof}

\begin{theorem}
For the lottery model, {\sc IsStabilityProbabilityNon-Zero} is polynomial-time solvable when each agent has at most two possible preference orderings.
\end{theorem}
\begin{proof}
The problem is to decide whether there is some preference ordering for each agent (among the ones in their lottery) such that the given matching is stable. If each agent has at most two possible preference orderings in their lottery, we can reduce the problem to an instance~$\varphi$ of 2SAT, as follows.

Let~$\{ a_1,\dotsc,a_n \}$ and~$\{ b_1,\dotsc,b_n \}$ be the two sets of agents. Moreover, for each agent~$c$ and each~$i \in \{ 1,2 \}$, let~$\mtext{pref}(c,i)$ denote the $i$-th preference in the lottery for agent~$c$.

We introduce a propositional variable for each preference~$\mtext{pref}(c,i)$---which we also call~$\mtext{pref}(c,i)$. Intuitively, these variables indicate which preference for the agents we choose to make the matching stable.

For each agent~$c$, we add the following clauses to~$\varphi$, to ensure that for each agent~$c$ there is exactly one preference that is selected:
$ (\mtext{pref}(c,1) \vee \mtext{pref}(c,2))
\quad\wedge\quad
(\neg \mtext{pref}(c,1) \vee \neg \mtext{pref}(c,2)). $

Then, we add clauses to ensure that the selected matching is stable. For each agent~$c$ and each~$i \in \{ 1,2 \}$, let~$B_{c,i}$ be the set of preferences~$\mtext{pref}(c',i')$---for~$c' \neq c$ and~$i' \in \{ 1,2 \}$---such that~$\mtext{pref}(c,i)$ and~$\mtext{pref}(c',i')$ together lead to the given matching being unstable (with $(c,c')$ being a blocking pair). Then, for each~$c,i$, we add the following clauses:
$ (\neg \mtext{pref}(c,i) \vee \neg \mtext{pref}(c',i'))
\quad \mtext{for each~$\mtext{pref}(c',i') \in B_{c,i}$.} $

The given matching is then stable if and only if~$\varphi$ is satisfiable. Since~$\varphi$ is a 2CNF, this can be decided in linear time.
\end{proof}

\begin{theorem}
For the lottery model, {\sc StabilityProbability} is \#P-complete, even when each agent has at most two possible preferences.
\end{theorem}
\begin{proof}
    We show how to count the number of satisfying assignments for
    a 2CNF formula using the problem {\sc StabilityProbability} for the
    lottery model where each agent has two possible preferences.
    Since this problem is \#P-hard, we get \#P-hardness
    also for {\sc StabilityProbability}.

    Let~$\varphi$ be a 2CNF formula over the variables~$x_1,\dotsc,x_n$.
    We firstly transform~$\varphi$ to a 2CNF formula~$\varphi'$ over the
    variables~$x_1,\dotsc,x_n,y_1,\dotsc,y_n$ that has exactly the same
    number of satisfying assignments, and that satisfies the property
    that each clause contains one variable among~$x_1,\dotsc,x_n$
    and one variable among~$y_1,\dotsc,y_n$.
    We do so as follows.
    Firstly, for each~$1 \leq i \leq n$, we add the
    clauses~$(\neg x_i \vee y_i)$ and~$(\neg y_i \vee x_i)$, ensuring that
    in each satisfying assignment the variables~$x_i$ and~$y_i$ get
    assigned the same truth value.
    Then, for each clause of~$\varphi$, we replace one occurrence of a
    variable among~$x_1,\dotsc,x_n$ by the corresponding variable
    among~$y_1,\dotsc,y_n$, and we add the resulting clause to~$\varphi'$.
    For example, if~$\varphi$ contains the clause~$(x_1 \vee \neg x_3)$,
    we would add the clause~$(x_1 \vee \neg y_3)$ to~$\varphi'$.
    It is readily verified that~$\varphi'$ has the same number of satisfying
    assignments as~$\varphi$.

    Moreover, we may assume without loss of generality that for
    any two variables of~$\varphi'$, there is at most one clause of~$\varphi'$
    that contains both of these variables.
    If in~$\varphi$ there are two variables~$x_1$ and~$x_2$
    and clauses~$(x_1 \vee x_2)$ and~$(\neg x_1 \vee \neg x_2)$,
    for instance, we can construct~$\varphi'$ to contain the
    clauses~$(x_1 \vee y_2)$ and~$(\neg y_1 \vee \neg x_2)$.

    We now construct an instance of {\sc StabilityProbability}.
    The sets of agents that we consider
    are~$\{ x_1,\dotsc,x_n, a_1,\dotsc,a_n \}$
    and~$\{ y_1,\dotsc,y_n,b_1,\dotsc,b_n \}$.
    The matching that we consider matches~$x_i$ to~$b_i$
    and matches~$y_i$ to~$a_i$, for each~$1 \leq i \leq n$.
    This is depicted below.
    Each agent~$b_i$ has only a single possible preference,
    namely one where they prefer~$x_i$ over all other agents.
    Similarly, each agent~$a_i$ has a single possible preference
    where they prefer~$y_i$ over all other agents.
    In other words, the agents~$a_i$ and~$b_i$ are perfectly
    happy with the given matching.
%\vspace{-0.5em}
\begin{figure}[h]
\centering
\begin{tikzpicture}[xscale=0.8,yscale=0.7]
  \node[label=above:{$b_1$}] (a1) at (0,0) {$\bullet$};
  \node[label=below:{$x_1$}] (b1) at (0,-1) {$\bullet$};
  \path[draw,-] (a1) -- (b1);
  \node[label=above:{$b_2$}] (a2) at (1,0) {$\bullet$};
  \node[label=below:{$x_2$}] (b2) at (1,-1) {$\bullet$};
  \path[draw,-] (a2) -- (b2);
  \node[label=above:{$b_3$}] (a3) at (2,0) {$\bullet$};
  \node[label=below:{$x_3$}] (b3) at (2,-1) {$\bullet$};
  \path[draw,-] (a3) -- (b3);
  \node[] at (3,-0.5) {$\dotsm$};
  \node[label=above:{$b_{n}$}] (an) at (4,0) {$\bullet$};
  \node[label=below:{$x_{n}$}] (bn) at (4,-1) {$\bullet$};
  \path[draw,-] (an) -- (bn);

  \node[label=above:{$y_1$}] (a1') at (6,0) {$\bullet$};
  \node[label=below:{$a_1$}] (b1') at (6,-1) {$\bullet$};
  \path[draw,-] (a1') -- (b1');
  \node[label=above:{$y_2$}] (a2') at (7,0) {$\bullet$};
  \node[label=below:{$a_2$}] (b2') at (7,-1) {$\bullet$};
  \path[draw,-] (a2') -- (b2');
  \node[label=above:{$y_3$}] (a3') at (8,0) {$\bullet$};
  \node[label=below:{$a_3$}] (b3') at (8,-1) {$\bullet$};
  \path[draw,-] (a3') -- (b3');
  \node[] at (9,-0.5) {$\dotsm$};
  \node[label=above:{$y_{n}$}] (an') at (10,0) {$\bullet$};
  \node[label=below:{$a_{n}$}] (bn') at (10,-1) {$\bullet$};
  \path[draw,-] (an') -- (bn');

  %
  % \begin{scope}[yshift=-1cm]
%   \node[label=above:{$y_1$}] (a1') at (0,-2.5) {$\bullet$};
%   \node[label=below:{$a_1$}] (b1') at (0,-3.5) {$\bullet$};
%   \path[draw,-] (a1') -- (b1');
%   \node[label=above:{$y_2$}] (a2') at (1,-2.5) {$\bullet$};
%   \node[label=below:{$a_2$}] (b2') at (1,-3.5) {$\bullet$};
%   \path[draw,-] (a2') -- (b2');
%   \node[label=above:{$y_3$}] (a3') at (2,-2.5) {$\bullet$};
%   \node[label=below:{$a_3$}] (b3') at (2,-3.5) {$\bullet$};
%   \path[draw,-] (a3') -- (b3');
%   \node[] at (3,-3) {$\dotsm$};
%   \node[label=above:{$y_{n}$}] (an') at (4,-2.5) {$\bullet$};
%   \node[label=below:{$a_{n}$}] (bn') at (4,-3.5) {$\bullet$};
%   \path[draw,-] (an') -- (bn');
 % \end{scope}
  %
\end{tikzpicture}
\label{fig:lottery-reduction-2}
%\caption{Illustration of the reduction}
\end{figure}
%\vspace{-1em}
    The agents~$x_i$ and~$y_i$ each have two possible preferences,
    that are each chosen with probability~$\frac{1}{2}$.
    These two possible preferences are associated with setting these
    variables to true or false, respectively.
    We describe how these preferences are constructed for the
    agents~$x_i$. The construction for the preferences of the
    agents~$y_i$ is then entirely analogous.

    Take an arbitrary agent~$x_i$.
    We show how to construct the two possible preferences for agent~$x_i$,
    which we denote by~$p_{x_i}$ and~$p_{\neg x_i}$.
    Both of these possible preferences are based on the following
    partial ranking:
    $ b_1 > b_2 > \dotsm > b_n, $
    and we add some of the agents~$y_1,\dotsc,y_n$ to the top of this
    partial ranking, and the remaining agents to the bottom of this partial
    ranking.

    To the ranking~$p_{x_i}$ we add exactly those agents~$y_j$ to the
    top where~$\varphi'$ contains a clause~$(\neg x_i \vee y_j)$
    or a clause~$(\neg x_i \vee \neg y_j)$.
    All remaining agents we add to the bottom.
    Similarly, to the ranking~$p_{\neg x_i}$ we add exactly those
    agents~$y_j$ to the top where~$\varphi'$ contains a
    clause~$(x_i \vee y_j)$ or a clause~$(x_i \vee \neg y_j)$.
The rankings~$p_{y_i}$ and~$p_{\neg y_i}$, for the agents~$y_i$,
    are constructed entirely similarly.

    Now consider a truth assignment~$\alpha :
    \{ x_1,\dotsc,x_n, y_1,\dotsc,y_n \} \rightarrow \{ 0,1 \}$,
    and consider the corresponding choice of preferences for the
    agents~$x_1,\dotsc,x_n,y_1,\dotsc,y_n$, where for each
    agent~$x_i$ the preference~$p_{x_i}$ is chosen if and only
    if~$\alpha(x_i) = 1$, and for each
    agent~$y_i$ the preference~$p_{y_i}$ is chosen if and only
    if~$\alpha(y_i) = 1$.
    Then~$\alpha$ satisfies~$\varphi'$ if and only if the corresponding
    choice of preferences leads to the matching being stable.
Since each combination of preferences is equally likely to occur,
    and there are~$2^{2n}$ many combinations of preferences,
    the probability that the given matching is stable is
    exactly~$q = \frac{s}{2^{2n}}$, where~$s$ is the number of satisfying
    truth assignments for~$\varphi$.
    Therefore, given~$q$,~$s$ can be obtained by
    computing~$s = q 2^{2n}$.
\end{proof}

\noindent
If each agent is allowed to have three possible preferences, then even the following problem is NP-complete. The statement can be proved via a reduction from Exact Cover by 3-Sets (X3C).

\begin{restatable}{theorem}{thProbStabNPC}\label{th:prob-stable-npc}
%For the lottery model, checking whether a given matching has non-zero probability of being stable is NP-complete.
For the lottery model, {\sc IsStabilityProbabilityNon-Zero} is NP-complete.
\end{restatable}
\begin{proof} %[Proof of Theorem~\ref{th:prob-stable-npc}]
   The problem is in NP, since we only need to provide one profile that occurs with non-zero probability for which the given matching is stable. We show NP-hardness by giving a reduction from Exact Cover by 3-Sets (X3C). Let~$(X,C)$ be an instance of X3C, where~$|X| = 3n$ for some~$n$, and~$C = \{ c_1,\dotsc,c_m \}$ is a collection of sets~$c_i \subseteq X$, each of size~$3$. Moreover, let~$c_i = \{ x_{\ell_{i,1}}, x_{\ell_{i,2}}, x_{\ell_{i,3}} \}$, for each~$1 \leq i \leq m$. The problem is to decide whether there is a subset~$C'\subseteq C$ of size exactly~$n$ such that~$\bigcup C' = X$.

   We construct an instance of our problem as follows. We let~$\{ a_1,\dotsc,a_{n}, \allowbreak{}a'_1,\dotsc,a'_{3n} \}$ and~$\{ b_1,\dotsc,b_{n},b'_1,\dotsc,b'_{3n} \}$ be the two sets of agents, we match~$a_i$ to~$b_i$---for each~$1 \leq i \leq n$---
   and we match~$a'_j$ to~$b'_j$---for each~$1 \leq j \leq 3n$. This is depicted below.

   \begin{figure}
   \centering
   \begin{tikzpicture}
     \centering
      \begin{scope}
     \node[label=above:{$a_1$}] (a1) at (0,0) {$\bullet$};
     \node[label=below:{$b_1$}] (b1) at (0,-1) {$\bullet$};
     \path[draw,-] (a1) -- (b1);
     \node[label=above:{$a_2$}] (a2) at (1,0) {$\bullet$};
     \node[label=below:{$b_2$}] (b2) at (1,-1) {$\bullet$};
     \path[draw,-] (a2) -- (b2);
     \node[label=above:{$a_3$}] (a3) at (2,0) {$\bullet$};
     \node[label=below:{$b_3$}] (b3) at (2,-1) {$\bullet$};
     \path[draw,-] (a3) -- (b3);
     \node[] at (3,-0.5) {$\dotsm$};
     \node[label=above:{$a_{n}$}] (an) at (4,0) {$\bullet$};
     \node[label=below:{$b_{n}$}] (bn) at (4,-1) {$\bullet$};
     \path[draw,-] (an) -- (bn);

     \node[label=above:{$a'_1$}] (a1') at (6,0) {$\bullet$};
     \node[label=below:{$b'_1$}] (b1') at (6,-1) {$\bullet$};
     \path[draw,-] (a1') -- (b1');
     \node[label=above:{$a'_2$}] (a2') at (7,0) {$\bullet$};
     \node[label=below:{$b'_2$}] (b2') at (7,-1) {$\bullet$};
     \path[draw,-] (a2') -- (b2');
     \node[label=above:{$a'_3$}] (a3') at (8,0) {$\bullet$};
     \node[label=below:{$b'_3$}] (b3') at (8,-1) {$\bullet$};
     \path[draw,-] (a3') -- (b3');
     \node[] at (9,-0.5) {$\dotsm$};
     \node[label=above:{$a'_{3n}$}] (an') at (10,0) {$\bullet$};
     \node[label=below:{$b'_{3n}$}] (bn') at (10,-1) {$\bullet$};
     \path[draw,-] (an') -- (bn');

       \end{scope}
      % \hspace{3cm}
     % \begin{scope}%[xshift=6cm]
     % \node[label=above:{$a'_1$}] (a1') at (-0.5,-2.5) {$\bullet$};
     % \node[label=below:{$b'_1$}] (b1') at (-0.5,-3.5) {$\bullet$};
     % \path[draw,-] (a1') -- (b1');
     % \node[label=above:{$a'_2$}] (a2') at (0.5,-2.5) {$\bullet$};
     % \node[label=below:{$b'_2$}] (b2') at (0.5,-3.5) {$\bullet$};
     % \path[draw,-] (a2') -- (b2');
     % \node[label=above:{$a'_3$}] (a3') at (1.5,-2.5) {$\bullet$};
     % \node[label=below:{$b'_3$}] (b3') at (1.5,-3.5) {$\bullet$};
     % \path[draw,-] (a3') -- (b3');
     % \node[] at (3,-3) {$\dotsm$};
     % \node[label=above:{$a'_{3n}$}] (an') at (4.5,-2.5) {$\bullet$};
     % \node[label=below:{$b'_{3n}$}] (bn') at (4.5,-3.5) {$\bullet$};
     % \path[draw,-] (an') -- (bn');
     % \end{scope}
     %
   \end{tikzpicture}
   \label{fig:lottery-reduction-1}
   \caption{Illustration of the reduction}
   \end{figure}
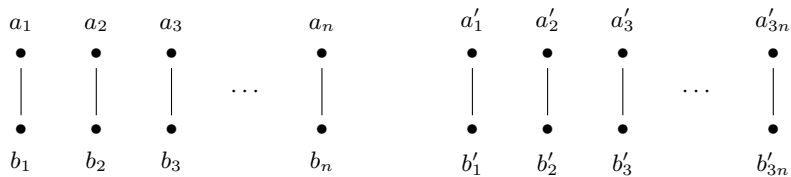

   Each agent~$a_i$ prefers their matching to~$b_i$ over any other possible match, i.e., agent~$a_i$ has one preference, where~$b_i$
   is ranked first, and the rest of the agents appear in some fixed order after $b_i$.

   Similarly, each agent~$b'_j$ prefers their matching to~$a'_j$ over any other possible match. That is, agent~$b'_j$ has one preference ordering in which ~$a'_j$ is ranked first and the rest of the agents appear in some fixed order after $a'_j$.

   Then, for each agent~$b_i$, we add the following $|C|$ possible preferences to the lottery:
   \begin{align*}
     P_{i,1} & : \quad a'_{\ell_{1,1}} > a'_{\ell_{1,2}} > a'_{\ell_{1,3}} > a_i > \dotsm \\
     P_{i,2} & : \quad a'_{\ell_{2,1}} > a'_{\ell_{2,2}} > a'_{\ell_{2,3}} > a_i > \dotsm \\
     & \hspace{1cm} \vdots \\
     P_{i,m} & : \quad a'_{\ell_{m,1}} > a'_{\ell_{m,2}} > a'_{\ell_{m,3}} > a_i > \dotsm
   \end{align*}
   where in each preference the remaining agents appear in any (fixed) order after~$a_i$. In other words,~$b_i$ prefers three agents~$a'_j$ to their current match, and these three form some set~$c \in C$.

   Finally, for each agent~$a'_j$, we add the following $n$ possible preferences to the lottery:
   \[ \begin{array}{r r}
     P'_{j,1} : \quad & b_2 > \dotsm > b_n > b'_j > b_1 > b'_1 > \dotsm > b'_{j-1} > b'_{j+1} > \dotsm > b'_{3n} \\
     P'_{j,2} : \quad & b_1 > b_3 > \dotsm > b_n > b'_j > b_2 > b'_1 > \dotsm > b'_{j-1} > b'_{j+1} > \dotsm > b'_{3n}\\
     P'_{j,3} : \quad & b_1 > b_2 > b_4 > \dotsm > b_n > b'_j > b_3 > b'_1 > \dotsm > b'_{j-1} > b'_{j+1} > \dotsm > b'_{3n} \\
     \vdots\quad\ & \vdots\qquad\quad\  \\
     P'_{j,n} : \quad & b_1 > \dotsm > b_{n-1} > b'_j > b_n > b'_1 > \dotsm > b'_{j-1} > b'_{j+1} > \dotsm > b'_{3n}
   \end{array} \]
   That is, each agent~$a'_j$ prefers each of the agents~$b_1,\dotsc,b_n$, except one, to their current match (and they never prefer any of the agents~$b'_{j'}$ for~$j' \neq j$ over their current match).
   %
   %An example of this construction for a concrete instance can be
   %found in Example~\ref{ex:prob-stable-npc}.

   We can show that there is a choice of preferences for the agents that
   makes this matching stable if and only if~$(X,C) \in \mtext{X3C}$.

   $(\Rightarrow)$
   Firstly, suppose that there is a choice of preferences for the agents that
   makes this matching stable.
   That is, for each agent~$b_i$ there is some preference
   ordering~$P_{i,\ell_i}$, and for each agent~$a'_j$ there is
   some preference ordering~$P'_{j,k_j}$, such that
   these orderings (together with the fixed preference orderings for the
   agents~$a_i$ and~$b'_j$) make this matching stable.
   Now, consider the set~$C' = \{ c_{\ell} : i \in [n], \ell = \ell_i \}$.
   We show that~$\bigcup C' = X$.
   To derive a contradiction, suppose that this is not the case,
   that is, suppose that~$\bigcup C' \neq X$.
   Then, since~$|C'| = n$,~$|X| = 3n$ and for each~$c \in C'$
   it holds that~$|c| = 3$, we know that there must be
   some~$c_{\ell},c_{\ell'} \in C'$ such that~$c_{\ell} \cap c_{\ell'}
   \neq \emptyset$.
   Say that~$x_j \in c_{\ell} \cap c_{\ell'}$.
   Therefore, there must be some~$i,i' \in [n]$ such
   that both~$b_i$ and~$b_{i'}$ prefer~$a'_{j}$ over
   their current match.
   On the other hand,~$a'_{j}$ will prefer either~$b_i$
   or~$b_{i'}$ over their current match.
   Therefore, either~$b_i$ and~$a'_{j}$
   or~$b_{i'}$ and~$a'_{j}$ will form a blocking pair.
   Thus, the matching is not stable.
   From this we can conclude that~$\bigcup C' = X$.

   $(\Leftarrow)$
   Conversely, suppose that there exists some~$C' \subseteq C$
   of size exactly~$n$ such that~$\bigcup C' = X$.
   Let~$C' = \{ c_{\ell_1},\dotsc,c_{\ell_n} \}$.
   Now, for each agent~$b_i$ we pick some preference
   ordering, and for each agent~$a'_j$ we pick
   some preference ordering, such that
   these orderings (together with the fixed preference orderings for the
   agents~$a_i$ and~$b'_j$) make the matching stable.
   For each agent~$b_i$, we pick the preference
   ordering~$P_{i,\ell_i}$, and for each agent~$a'_j$
   we pick the preference ordering~$P'_{j,k_j}$,
   where~$k_j \in [n]$ is the unique value such
   that~$x_{j} \in c_{\ell_{k_j}}$.
   It is straightforward to verify that these preferences make
   the matching stable.
   \end{proof}

%The same reduction can also be used to prove that for the the lottery model,  {\sc StabilityProbability} is \#P-complete.

\noindent
% We also obtain the following corollary from Theorem~\ref{th:prob-stable-npc}.

We obtain the first corollary from Theorem~\ref{th:prob-stable-npc} and the second from \cite[Proposition~8]{WeMi00a} and Theorem~\ref{th:prob-stable-npc}.

\begin{corollary}\label{cor:approx}
For the lottery model, unless $\text{P}=\text{NP}$, there exists no polynomial-time algorithm for approximating {\sc StabilityProbability} of a given matching.
\end{corollary}

% Since {\sc StabilityProbability} is \#P-complete, one may wonder whether it admits an \emph{FPRAS (fully polynomial time randomized approximation scheme)} or not.
% Let $\Gamma$ be a finite alphabet in which we agree to describe our problem instances and solutions.
% A \emph{fully polynomial time randomized approximation scheme (FPRAS)} for a function $f : \Gamma^* \beats \mathbb{Q}$ is an algorithm that takes input $x \in \Gamma^*$ and a parameter $\epsilon \in \mathbb{Q}_{> 0}$, and returns with probability at least $\frac34$ a number between $f(x)/(1+\epsilon)$ and $(1+\epsilon) f(x)$. Moreover, an FPRAS runs in time polynomial in the size of $x$ and $\nicefrac{1}{\epsilon}$.
% RP is the complexity class consisting of problems that can be solved in randomized polynomial time.

 % \noindent
 % The next statement follows from \cite[Proposition~8]{WeMi00a} and Theorem~\ref{th:prob-stable-npc}.

\begin{corollary}\label{cor:fpras}
For the lottery model, unless $\text{NP}=\text{RP}$, there is no FPRAS for {\sc StabilityProbability}.
\end{corollary}

%%%%%%%%%%%%%%%%%%%%%%%%%%%%%%%%%%%%%%%%%
%
%%%%%%%%%%%%%%%%%%%%%%%%%%%%%%%%%%%%%%%%%
\section{Compact Indifference Model}
%In this section we focus on compact indifference model.
The compact indifference model is equivalent to assuming that we are given an instance of SMT and each linear order over candidates (each possible preference ordering) is achieved by breaking ties independently at random with uniform probabilities.
It is easy to show that {\sc IsStabilityProbablityNonZero}, {\sc IsStabilityProbablityOne}, and {\sc ExistsCertainlyStableMatching} are all in P. 
%The corresponding claims and the proof can be found in \cite{}.

\begin{proposition}\label{IsStabilityProbabilityNonZero-compact}
For the compact indifference model, {\sc IsStabilityProbabilityNonZero} is in P.
\end{proposition}
\begin{proof}
This is equivalent to checking whether a given matching $\mu$ is weakly stable in the given SMTI instance. To check this we only have to look for a blocking pair, which can be done in polynomial time: take every possible pair $(m,w)$ who are not matched together and check whether they both strictly prefer each other to their current partner.
\end{proof}

\begin{proposition}\label{IsStabilityProbabilityOne-compact}
For the compact indifference model, {\sc IsStabilityProbabilityOne} is in P.
\end{proposition}
\begin{proof}
The problem is polynomial-time solvable. We go through all the blocking pairs and check if any blocking pair is feasible. For each blocking pair, we break ties (if there are any) in favour of the blocking pair. Given that we break ties in favour of the blocking pairs, if there exists a blocking pair that is feasible, the stability probability is not one.
% {\sc IsStabilityProbabilityOne} is in P for the lottery model, which is a generalization of compact indifference.
% \haris{P-time result for lottery does not imply P-time result for compact since the reduction is not polynomial-time/size. We just need to refer the statement that verifying superstable matching is in P or to write an argument that check for each possible blocking pair after reinforcing that blocking pair by breaking ties in favor of the blocking pair. }
% \nick{agreed with haris, this algorithm isn't poly space/time...}
%I've written down the explicit argument.
\end{proof}

\begin{proposition}\label{IsCertainlyStableMatching-compact}
For the compact indifference model, {\sc ExistsCertainlyStableMatching} is in P.
\end{proposition}
\begin{proof}
Deciding whether there is matching that is stable with probability one is equivalent to deciding whether there is a matching that is stable w.r.t. all refinements, a super-stable matching. Given an instance of SMTI one can decide in polynomial time whether it admits a super-stable matching or not \cite{manlove99}.
\end{proof}

We do not yet know the complexity of computing the stability probability of a given matching under the compact indifference model, but this problem can be shown to be in P if one side has certain preferences.
\begin{theorem}\label{StabilityProbability-compact}
In the compact indifference model, if one side has certain preferences, {\sc StabilityProbability} is polynomial-time solvable.
\end{theorem}
\begin{proof}
Assume, w.l.o.g., that men have certain preferences. The following procedure gives us the stability probability of any given matching $\mu$. (1) For each uncertain woman $w$ identify those men with whom she can potentially form a blocking pair. That is, those $m$ such that $w \pref_{m} \mu(m)$ and $w$ is indifferent between $m$ and her partner in $\mu$. Assume there are $k$ of such men. The probability of $w$ not forming a blocking pair with any men is then $\frac{1}{k+1}$. (2) Multiply the probabilities from step 1.
\end{proof}

%\bahar{If we end up keeping the assumption of complete lists, we should perhaps replace SMTI with SMT throughout this section and in some proofs in the appendix. Though we should probably remark that all these results hold even if we assume incomplete lists.}
%peter2: yes, I did that
\noindent
We next show that {\sc MatchingWithHighestStabilityProbability} is NP-hard.
%most stable matching \nick{\textsc{MatchWithHighestStabilityProbability???} Should use the right name..} is hard.
%haris: done.
%The problem is to find a matching $\mu$ that is stable in most of the derived instances (and therefore it is the most stable matching in expectation).
For an instance $I$ of SMT and matching $\mu$, let $p(\mu,I)$ denote the probability of $\mu$ being stable, and let $p_S(I)=max\{p(\mu,I) | \mu \mbox{ is a matching in } I\}$, that is the maximum probability of a matching being stable. %Furthermore we say that $\mu$ is one of the most stable matchings if its probability of being stable is maximal among all matchings.
%First we note some simple facts. %on the nature of this problem.
A matching $\mu$ is said to be weakly stable if there exists a tie-breaking rule where $\mu$ is stable. Therefore a matching $\mu$ has positive probability of being stable if and only if it is weakly stable. Furthermore, if the number of possible tie-breaking is $N$ then any weakly stable matching has a probability of being stable at least $\frac{1}{N}$.

An extreme case occurs if we have one woman only with $n$ men, where the woman is indifferent between all men. In this case any matching (pair) has a $\frac{1}{n}$ probability of being stable. An even more unfortunate scenario is when we have $n$ men and $n$ women, each women is indifferent between all men, and each man ranks the women in a strict order in the same way, e.g. in the order of their indices. In this case, the probability that the first woman picks her best partner, and thus does not block any matching is $\frac{1}{n}$. Suppose that the first woman picked her best partner, the probability that the second woman also picks her best partner from the remaining $n-1$ men is $\frac{1}{n-1}$, and so on. Therefore, the probability that an arbitrary complete matching is stable is $\frac{1}{n(n-1)\dots 2}=\frac{1}{n!}$.

\begin{restatable}{theorem}{thmHighestStabprobCompact}\label{MatchingWithHighestStabilityProbability-compact}
For the compact indifference model {\sc MatchingWithHighestStabilityProbability} is NP-hard, even if only one side of the market has uncertain agents.
\end{restatable}
\begin{proof} %[of Theorem~\ref{MatchingWithHighestStabilityProbability-compact}]
                    %    \haris{Corollary 3.2 or Corollary 3.4? In the paper 3.2 is a theorem not a corollary.}
                    %peter: Corollary 3.4, thanks for spotting!
For an instance $I$ of SMTI, let $opt(I)$ denote the maximum size of a weakly stable matching in $I$. Halldorsson et al. \cite{HIY07a} showed [in the proof of Corollary 3.4] that given an instance $I$ of SMTI of size $n$, where only one side of the market has agents with indifferences and each of these agents has a single tie of size two, and any arbitrary small positive $\epsilon$, it is NP-hard to distinguish between the following two cases:
\begin{inparaenum}[(1)]
      \item $opt(I) \geq \frac{21-\epsilon}{27}n$ and
      \item $opt(I) < \frac{19+\epsilon}{27}n$.
\end{inparaenum}

When choosing $\epsilon$ so that $0<\epsilon <\frac{1}{2}$ we can simplify the above cases to
\begin{inparaenum}[(1)]
      \item $opt(I) > \frac{41}{54}n$, since $opt(I) \geq \frac{21-\epsilon}{27}n > \frac{41}{54}n$ and
      \item $opt(I) < \frac{39}{54}n$, since $opt(I) < \frac{19+\epsilon}{27}n < \frac{39}{54}n$.
\end{inparaenum}

Therefore, the number of agents left unmatched on either side of the market is less than $\frac{13}{54}n$ in the first case and more than $\frac{15}{54}n$ in the second case. Let us now extend instance $I$ to a larger instance of SMTI $I'$ as follows. Besides the $n$ men $M=\{m_1, \dots , m_n\}$ and $n$ women $W=\{w_1, \dots ,w_n\}$, we introduce $\frac{13}{54}n$ men $X=\{x_1,\dots x_k\}$ and another $\frac{n}{27}$ men $Y=\{y_1,\dots y_l\}$ and $\frac{n}{27}$ women $Z=\{z_1,\dots z_l\}$. Furthermore, for each $y_j\in Y$, we introduce $n$ men $Y^j=\{y^j_1,\dots , y^j_n\}$. We create the preferences of $I'$ as follows. The preferences of men $M$ remain the same. For each woman $w\in W$ we append the men $X$ and then $Y$ at the end of her list in the order of their indices. Each man $x_i\in X$ has only all the women $W$ in his list in the order of their indices. Furthermore, each $y_j\in Y$ has all the women $W$ first in his preference list in the order or their indices and then $z_j$. Let each $z_j\in Z$ has $y_j$ as first choice and then all the men $Y^j$ in one tie of size $n$. Each man in $Y^j$ has only $z_j$ in his list. We will show that in case one $p_S(I')\geq \frac{1}{2^n}$, whilst in case two $p_S\leq (\frac{1}{n})^{\frac{n}{27}}$. Therefore, for $n>2^{27}$, it is NP-hard to decide which of the two separate intervals contains the value $p_S(I')$.

                    To show the above statement, suppose first that we have the first case, so $opt(I) > \frac{41}{54}n$ and therefore less than $\frac{13}{54}n$ women are left unmatched in a maximum size weakly stable matching $\mu$ for $I$, denoted by $W_u\subset W$. We extend $\mu$ to $\mu'$ for $I'$ as follows. We assign all the women in $W_u$ to men in $X$ in the unique stable way, namely we pair them in a mutually increasing order of their indices. Since $|X|>|W_u|$, we now matched all women in $W$, and left some men in $X$ unmatched in $\mu'$. We complete the matching by assigning $y_j$ to $z_j$ for each $j=1,\dots ,n$ and leaving all of the men in $Y^j$ for all $j$ unmatched. We shall see that no matter how we break the ties in $I'$, blocking pair can appear between the original $I$ agents only, and therefore the probability of $\mu'$ being stable in $I'$ is the same as the probability of $\mu$ being stable in $I$. Since we have at most $n$ ties in $I$, each of length two, the number of different tie-breakings is at most $2^n$, out of which at least one is stable. Therefore $p(\mu,I')=p(\mu,I)\geq \frac{1}{2^n}$.

                    In the second case, $opt(I) < \frac{39}{54}n$ and therefore more than $\frac{15}{54}n$ women are left unmatched in any weakly stable matching $\mu$ for $I$. Let $\mu'$ be one of the most stable matchings in $I'$. First we have to note that the restriction of $\mu'$ to $I$ must be weakly stable in $I$, since otherwise $p(\mu',I')=0$. Let $W_u$ denote the set of women that are not matched to any man from $M$ in $\mu'$. According to our assumption $|W_u|>\frac{15}{54}n$, whilst $|X|+|Y|=\frac{15}{54}n$, therefore in order to avoid a certain blocking pair between $W_u$ and $X\cup Y$ we shall match all the men in $X\cup Y$ to women in $W_u$ in the only stable way (in the order of indices, where men in $X$ are coming before men in $Y$), an leaving some women in $W_u$ unmatched in $\mu'$. However, in this case no agent $z_j\in Z$ can be matched to $y_j$, and therefore, even if there was no potential blocking pair between agents of $I$, the probability that $z_j$ is matched the best partner from $Y^j$ is $\frac{1}{n}$ independently for each $z_j\in Z$. Therefore the probability of $\mu'$ being stable is at most $(\frac{1}{n})^{\frac{n}{27}}$, which completes the proof of the first statement.

                    Regarding the NP-hardness of finding one of the most stable matchings, we shall prove that we can decide between the two cases according to the number of unmatched women in $W$ in the restriction of $\mu'$ to $I$, where $\mu'$ is one of the most stable matchings in $I'$. To see this, let $W_u$ denote again the set of women that are not matched to any man in $M$ under $\mu'$. In the first case, when $opt(I) > \frac{41}{54}n$, it must be the case that $|W_u|<\frac{15}{54}n$, since otherwise $p(\mu',I)$ would be less than $(\frac{1}{n})^{\frac{n}{27}}$ and could not achieve $\frac{1}{2^n}$, that is the minimum value for $p_S(I')$, as shown in the above argument. Whilst, in the second case $|W_u|>\frac{15}{54}n$ must hold, since $opt(I) < \frac{39}{54}n$ was assumed.%\haris{explain why it is obvious.} peter: added a short explanation
\end{proof}
\section{Joint Probability Model}

In this section, we examine  problems concerning the joint probability model.

% In this section, we examine computational problems concerning the joint probability model.

          % \textbf{Joint Probability Model}: a probability distribution over preference profiles

\begin{restatable}{theorem}{thmJointPoly}\label{thm:joint-stabilityprobability-poly}
For the joint probability model, {\sc StabilityProbability} can be solved in polynomial time.
\end{restatable}
\begin{proof} %[of Theorem ~\ref{thm:joint-stabilityprobability-poly}]
The probability that a given matching is stable is equivalent to the probability weight of the preference profiles for which the matching is stable. This can be checked as follows. We check the preference profiles for which the given matching is stable (for one profile, this can be checked in $O(n^2)$). Then we add the probabilities of those profiles for which the matching is stable. The sum of the probabilities is the probability that the matching is stable.
\end{proof}

\begin{corollary}
For the joint probability model, {\sc IsStabilityProbabilityNon-Zero} and {\sc IsStabilityProbabilityOne}  can be solved in polynomial time.
\end{corollary}

%          Observe that {\sc ExistsCertainlyStableMatching} is trivial: compute a stable matching for any of the possible realizations of the preference profiles.

%Bahar: removed the next line as we've already have a definition of the problem in the introduction
%{\sc ExistsCertainlyStableMatching}: does there exists a matching that is stable for all preference profiles.

              % \nick{Haris, now I am super confused at what is going on here... we have ExistsCertainlyStableMatching being easy in this sentence and then being NP-hard two sentences later..  this is what threw me for a loop before in my last long email.  I still don't understand how we check the intersection of the sets of all stable matchings without enumerating all stable matchings for each deterministic realization... can you explain this to me again?}
  %             \haris{I think I told you on chat that the statements in Corollary 4 and Theorem 11 are correct and the sentence now commented out had a typo (remnant from a previous version). For your previous email, the explanation was that our previous theorem (theorem 1 or 2) was for all uncertainty models and now we have realized that it is not for all models but only independent models. }

		  \noindent
          For the joint probability model,  the problem  {\sc ExistsCertainlyStableMatching} is equivalent to checking whether the intersection of the sets of stable matchings of the different preference profiles is empty or not.

\begin{restatable}{theorem}{thmJointExistsstable}
\label{thm:joint-model-exists-certainly-stable}
For the joint probability model, {\sc Exists\-Certainly\-Stable\-Matching} is NP-complete.
\end{restatable}
\begin{proof} %[sketch]
%PB: I modified the statement and added the following line in the proof:
The problem is in NP, since computing {\sc StabilityProbability} can be done in polynomial time by Theorem \ref{thm:joint-stabilityprobability-poly}.
The NP-hardness proof is by reduction from 3-Colorability.
Let~$G = (V,E)$ be a graph specifying an instance of 3-Colorability,
where~$V = \{ v_1,\dotsc,v_n \}$.
%
%\nick{ the next two sentences end is "as follows" which is a bit linguistically weird...}
We construct an instance
%of {\sc ExistsCertainlyStableMatching} for the joint probability model.
$I$ of SMULP assuming the joint probability model.

%The general idea behind this reduction is as follows.
For each vertex~$v_i \in V$, we introduce three men~$m_{i,1},m_{i,2},m_{i,3}$
and three women~$w_{i,1},w_{i,2},w_{i,3}$.
Then, we introduce one preference profile~$P_0$ that ensures that
every certainly stable matching matches---for each~$i \in [n]$---%
each~$m_{i,j}$ to some~$w_{i,j'}$ and, vice versa, each~$w_{i,j}$
to some~$m_{i,j'}$, for~$j,j' \in [3]$.
Moreover, it ensures that for each~$i \in [n]$, exactly one of three
matchings between the men~$m_{i,j}$ and the women~$w_{i,j}$
must be used:
\begin{small}
\[ \begin{array}{r l}
  \mtext{(1)\hspace{-6pt}} &
  \mtext{$m_{i,1}$ is matched to~$w_{i,1}$,
  $m_{i,2}$ is matched to~$w_{i,2}$, and
  $m_{i,3}$ is matched to~$w_{i,3}$;} \\
  \mtext{(2)\hspace{-6pt}} &
  \mtext{$m_{i,1}$ is matched to~$w_{i,2}$,
  $m_{i,2}$ is matched to~$w_{i,3}$, and
  $m_{i,3}$ is matched to~$w_{i,1}$; or} \\
  \mtext{(3)\hspace{-6pt}} &
  \mtext{$m_{i,1}$ is matched to~$w_{i,3}$,
  $m_{i,2}$ is matched to~$w_{i,1}$, and
  $m_{i,3}$ is matched to~$w_{i,2}$;} \\
\end{array} \]
\end{small}

\noindent
Intuitively, choosing one of the matchings~(1)--(3)
for the agents~$m_{i,j},w_{i,j}$ corresponds to coloring
vertex~$v_i$ with one of the three colors in~$\{ 1,2,3 \}$.

Then, for each edge~$e = \{ v_{i_1},v_{i_2} \} \in E$,
and for each color~$c \in \{ 1,2,3 \}$,
we introduce a preference profile~$P_{e,c}$ that
ensures that in any certainly stable matching,
the agents~$m_{i_1,j},w_{i_1,j}$ and
the agents~$m_{i_2,j},w_{i_2,j}$ cannot both
be matched to each other with matching~$(c)$.
We let each preference profile appear
with non-zero probability (e.g., we take a uniform lottery over
all the preference profiles that we introduced).
As a result, any certainly stable matching directly corresponds
to a proper 3-coloring of~$G$.

A detailed description of the preference profiles~$P_0$
and~$P_{e,c}$ and a proof of correctness for this reduction follows.

%We specify the preference profiles~$P_0$ and~$P_{e,c}$ in detail.
%We begin with~$P_0$.
In~$P_0$, for each~$i \in [n]$,
the preferences for~$m_{i,j},w_{i,j}$ are as follows:
% \begin{align*}
%   m_{i,1} & : \quad w_{i,1},w_{i,2},w_{i,3},---\\
%   m_{i,2} & : \quad w_{i,2},w_{i,3},w_{i,1},---\\
%   m_{i,3} & : \quad w_{i,3},w_{i,1},w_{i,2},---\\
%   w_{i,1} & : \quad m_{i,2},m_{i,3},m_{i,1},---\\
%   w_{i,2} & : \quad m_{i,3},m_{i,1},m_{i,2},---\\
%   w_{i,3} & : \quad m_{i,1},m_{i,2},m_{i,3},---
% \end{align*}

\begin{align*}
  m_{i,1} & : \quad w_{i,1},w_{i,2},w_{i,3},---&w_{i,1} & : \quad m_{i,2},m_{i,3},m_{i,1},--- \\
 m_{i,2} & : \quad w_{i,2},w_{i,3},w_{i,1},---&w_{i,2} & : \quad m_{i,3},m_{i,1},m_{i,2},---\\
  m_{i,3} & : \quad w_{i,3},w_{i,1},w_{i,2},---&  w_{i,3} & : \quad m_{i,1},m_{i,2},m_{i,3},---
  \end{align*}

%\bahar{I propose putting the above preferences in two columns}
%\haris{done:)}
Next, we continue with the preference profiles~$P_{e,c}$.
Take an arbitrary~$e = \{ v_{i_1},v_{i_2} \} \in E$
and an arbitrary~$c \in \{ 1,2,3 \}$.
In~$P_{e,c}$, the preferences for~$m_{i,j},w_{i,j}$
for each~$i \in [n] \setminus \{ i_1,i_2 \}$ are exactly the
same as in~$P_0$.
Only the preferences for~$m_{i_1,j},w_{i_1,j}$
and~$m_{i_2,j},w_{i_2,j}$ differ from~$P_0$;
namely, we construct these preferences as follows.

For~$m_{i_1,j},w_{i_1,j}$, we start with preferences that
(i)~for all~$m_{i_1,j}$ have~$w_{i_1,1},w_{i_1,2},w_{i_1,3}$ as top
three choices,
(ii)~for all~$w_{i_1,j}$ have~$m_{i_1,1},m_{i_1,2},m_{i_1,3}$ as top
three choices,
(iii)~admit only matchings~(1),~(2), and~(3) as stable matchings
between the agents~$m_{i_1,j},w_{i_1,j}$,
and (iv)~for the men~$m_{i_1,j}$ the matching~$(c)$ is the
worst option among the matchings~(1),~(2), and~(3).
Similarly, for~$m_{i_2,j},w_{i_2,j}$, we start with preferences that
satisfy conditions~(i),~(ii) and~(iii), and additionally satisfy
the condition~(iv$'$) that for the women~$w_{i_2,j}$ the matching~$(c)$ is the
worst option among the matchings~(1),~(2), and~(3).
Then, we modify the preferences for~$m_{i_1,1}$ and~$w_{i_2,1}$ slightly.
For~$m_{i_1,1}$, we insert~$w_{i_2,1}$ between his second and third
preferred woman.
Similarly, for~$w_{i_2,1}$, we insert~$m_{i_1,1}$ between her second and
third preferred man.
As a result,~$m_{i_1,1}$ and~$w_{i_2,1}$ form a blocking pair in this
preference profile if both the agents~$m_{i_1,j},w_{i_1,j}$
and the agents~$m_{i_2,j},w_{i_2,j}$ are matched to each other
using matching~$(c)$---%
and not if either set of agents is matched to each other
using some other matching~$(c')$.

%\bahar{We get a white big space here, I'm not sure why!}

For example, consider~$e = \{ v_{i_1},v_{i_2} \}$
and~$c = 2$.
The preferences for the agents~$m_{i_1,j},w_{i_1,j}$
and~$m_{i_2,j},w_{i_2,j}$ in the preference profile~$P_{e,c}$
are as follows:
% \begin{align*}
%   m_{i_1,1} & : \quad w_{i_1,1},w_{i_1,3},\bm{w_{i_2,1}},w_{i_1,2},---\\
%   m_{i_1,2} & : \quad w_{i_1,2},w_{i_1,1},w_{i_1,3},---\\
%   m_{i_1,3} & : \quad w_{i_1,3},w_{i_1,2},w_{i_1,1},---\\
%   w_{i_1,1} & : \quad m_{i_1,3},m_{i_1,2},m_{i_1,1},---\\
%   w_{i_1,2} & : \quad m_{i_1,1},m_{i_1,3},m_{i_1,2},---\\
%   w_{i_1,3} & : \quad m_{i_1,2},m_{i_1,1},m_{i_1,3},---\\
%   \\[-10pt]
%   m_{i_2,1} & : \quad w_{i_2,2},w_{i_2,3},w_{i_2,1},---\\
%   m_{i_2,2} & : \quad w_{i_2,3},w_{i_2,1},w_{i_2,2},---\\
%   m_{i_2,3} & : \quad w_{i_2,1},w_{i_2,2},w_{i_2,3},---\\
%   w_{i_2,1} & : \quad m_{i_2,1},m_{i_2,2},\bm{m_{i_1,1}},m_{i_2,3},---\\
%   w_{i_2,2} & : \quad m_{i_2,2},m_{i_2,3},m_{i_2,1},---\\
%   w_{i_2,3} & : \quad m_{i_2,3},m_{i_2,1},m_{i_2,2},---
% \end{align*}
%
\begin{small}
\begin{align*}
  m_{i_1,1} & : \quad w_{i_1,1},w_{i_1,3},\bm{w_{i_2,1}},w_{i_1,2},---&m_{i_2,1} & : \quad w_{i_2,2},w_{i_2,3},w_{i_2,1},---\\
  m_{i_1,2} & : \quad w_{i_1,2},w_{i_1,1},w_{i_1,3},---& m_{i_2,2} & : \quad w_{i_2,3},w_{i_2,1},w_{i_2,2},---\\
  m_{i_1,3} & : \quad w_{i_1,3},w_{i_1,2},w_{i_1,1},---&  m_{i_2,3} & : \quad w_{i_2,1},w_{i_2,2},w_{i_2,3},---\\
  w_{i_1,1} & : \quad m_{i_1,3},m_{i_1,2},m_{i_1,1},---&  w_{i_2,1} & : \quad m_{i_2,1},m_{i_2,2},\bm{m_{i_1,1}},m_{i_2,3},---\\
  w_{i_1,2} & : \quad m_{i_1,1},m_{i_1,3},m_{i_1,2},---&  w_{i_2,2} & : \quad m_{i_2,2},m_{i_2,3},m_{i_2,1},---\\
  w_{i_1,3} & : \quad m_{i_1,2},m_{i_1,1},m_{i_1,3},---&  w_{i_2,3} & : \quad m_{i_2,3},m_{i_2,1},m_{i_2,2},---
 % \\[-10pt]
  \end{align*}%
\end{small}%
We argue that~$G$ has a proper 3-coloring if and
only if there is a certainly stable matching for the probability distribution
over preference profiles that we constructed.

    $(\Rightarrow)$
    Firstly, suppose that~$G$ has a proper 3-coloring,
    say~$\chi : V \rightarrow \{ 1,2,3 \}$.
    We can then construct a certainly stable matching
    as follows.
    For each~$i \in [n]$, we
    match the agents~$m_{i,j},w_{i,j}$ to each other using
    matching~$(c_i)$, where~$c_i = \chi(v_i)$.
    Clearly, this matching is stable for~$P_0$.
    Moreover, because~$\chi$ is a proper 3-coloring of~$G$,
    it is straightforward to verify that this matching is also
    stable for each~$P_{e,c}$.

    $(\Leftarrow)$
    Conversely, suppose that there is a certainly stable matching.
    We know that in this matching, each man~$m_{i,j}$ must
    be matched to some woman~$w_{i,j'}$, and vice versa,
    each woman~$w_{i,j}$ must be matched to some man~$m_{i,j'}$.
    If this were not the case, the matching would not be stable
    for~$P_0$, and thus not certainly stable.
    Moreover, by a similar argument, we know that for
    each~$i \in [n]$, the matching between the men~$m_{i,j}$
    and the women~$w_{i,j}$ must be one of
    the matchings~(1),~(2), or~(3).
    We can then construct a 3-coloring~$\chi : V \rightarrow \{ 1,2,3 \}$
    as follows.
    For each~$i \in [n]$, we let~$\chi(v_i) = c_i$,
    where~$(c_i)$ is the matching used in the certainly stable matching
    to match the men~$m_{i,j}$ and the women~$w_{i,j}$ to each other.

    We argue that~$\chi$ is a proper 3-coloring of~$G$.
    Suppose that this is not the case,
    that is, that there is some~$e = \{ v_{i_1},v_{i_2} \}$
    such that~$\chi(v_{i_1}) = \chi(v_{i_2}) = c$.
    Now consider the preference profile~$P_{e,c}$.
    By construction of~$\chi$, we know that in the certainly stable
    matching, both the agents~$m_{i_1,j},w_{i_1,j}$ and
    the agents~$m_{i_2,j},w_{i_2,j}$ are matched to each other
    using matching~$(c)$.
    However, then by construction of~$P_{e,c}$,%
    ~$m_{i_1,1}$ and~$w_{i_2,1}$ form a blocking pair in~$P_{e,c}$.
    This is a contradiction with our assumption that the matching
    we considered is certainly stable.
    From this, we can conclude that~$\chi$ is a proper 3-coloring of~$G$.
\end{proof}

\noindent
% By modifying the proof of Theorem~\ref{thm:joint-model-exists-certainly-stable}, the following corollary can also be proved.
By modifying the proof of Theorem~\ref{thm:joint-model-exists-certainly-stable}, the following can also be proved.

\begin{restatable}{corollary}{corVizing}\label{cor:vizing}
For the joint probability model, {\sc Exists\-Certainly\-Stable\-Matching} is NP-complete, even when there are only 16 preference profiles in the lottery.
\end{restatable}
\begin{proof} %[of Corollary~\ref{cor:vizing}]
We show this by modifying the proof of
Theorem~\ref{thm:joint-model-exists-certainly-stable}.
We know that 3-Colorability is NP-hard already when restricted
to graphs of degree~4 \cite{Dailey80}.
We use the reduction in the proof of
Theorem~\ref{thm:joint-model-exists-certainly-stable}, and we
assume that the given graph~$G$ has degree~4.
Then, by Vizing's Theorem \cite{Vizing64},
we know that we can give a proper edge coloring of~$G$
that uses at most~5 colors.
Moreover, we can find such an edge coloring in polynomial time.
Then, since in the proof of
Theorem~\ref{thm:joint-model-exists-certainly-stable},
in each preference profile~$P_{e,c}$ with~$e = \{ v_{i_1},v_{i_2} \}$,
only the preferences for
the agents~$m_{i_1,j},w_{i_1,j},m_{i_2,j},w_{i_2,j}$ differ
from~$P_0$,
we can, for each color~$c \in \{ 1,2,3 \}$,
combine the preference profiles~$P_{e,c}$ for all
edges~$e$ that are colored with the same color.
This results in only~16 preference profiles:~$P_0$,
and a preference profile for each of the~5 edge colors
and each of the~3 vertex colors.
\end{proof}

%%%%%%%%%%%%%%%%%%%%%%%%%%%%
\section{Future work}

%peter2: I commented out the next two lines as we have already addressed this point in the text. I added some other questions instead...
%We considered a simple setting in which agents have strict preferences. All of our hardness results extend to more general settings with incomplete lists. Our algorithmic results can be straightforwardly extended to the settings with incomplete lists where some partners are considered unacceptable.
First we note that we left open two outstanding questions, as described in Table \ref{table:summary:uncertainmatch}. In this paper we focused on the problem of computing a matching with the highest stability probability. However, a similarly reasonable goal could be to minimize the expected number of blocking pairs. It would also be interesting to investigate some further realistic probability models, such as the situation when the candidates are ranked according to some noisy scores (like the SAT scores in the US college admissions). This would be a special case of the joint probability model that may turn out to be easier to solve. Finally, in a follow-up paper we are planning to investigate another probabilistic model that is based on independent pairwise comparisons.
%We expect many of our algorithmic results to

%For the hard problems, a special case to think about, which might be tractable: The case when only one agent is uncertain. Suppose that we have an instance $I$of SMTI where only one woman is uncertain, i.e., has some ties in her preference list.

% \begin{itemize}
%     \item Extending the (polynomial time) results to the case with incomplete preferences (\bahar{I think most of it is going to be straightforward}.
% \end{itemize}

%peter2: As Serge suggested I remove now the Acknowledgements, but will put back, if the paper is accepted ;-)
\paragraph{Acknowledgments.}
%\serge{Acknowledgments are only needed in the camera-ready version, not in the submitted version.}
Bir\'o is supported by the Hungarian Academy of Sciences under its Momentum Programme (LP2016-3) and the Hungarian Scientific Research Fund, OTKA, Grant No.\ K108673. Rastegari was supported EPSRC grant EP/K010042/1 at the time of the submission. The authors gratefully acknowledge the support from European Cooperation in Science and Technology (COST) action IC1205.
Serge Gaspers is the recipient of an Australian Research Council (ARC) Future Fellowship (FT140100048) and acknowledges support under the ARC's Discovery Projects funding scheme (DP150101134). NICTA is funded by the Australian Government through the Department of Communications and the ARC through the ICT Centre of Excellence Program.
%
%Ronald de Haan is supported by the Austrian Science Fund (FWF),
%project P26200 (Parameterized Compilation).

\renewcommand*{\bibfont}{\small}
\bibliography{abb,bibfile}

\begin{thebibliography}{16}
\providecommand{\natexlab}[1]{#1}
\providecommand{\url}[1]{\texttt{#1}}
\expandafter\ifx\csname urlstyle\endcsname\relax
  \providecommand{\doi}[1]{doi: #1}\else
  \providecommand{\doi}{doi: \begingroup \urlstyle{rm}\Url}\fi

\bibitem[Dailey(1980)]{Dailey80}
D.~P. Dailey.
\newblock Uniqueness of colorability and colorability of planar 4-regular
  graphs are {NP}-complete.
\newblock \emph{Discrete Mathematics}, 30\penalty0 (3):\penalty0 289--293,
  1980.

\bibitem[Drummond and Boutilier(2014)]{DrBo14a}
J.~Drummond and C.~Boutilier.
\newblock Preference elicitation and interview minimization in stable
  matchings.
\newblock In \emph{Proceedings of the Twenty-Eighth {AAAI} Conference on
  Artificial Intelligence}, pages 645--653, 2014.

\bibitem[Ehlers and Mass\'o(2015)]{EhlersMasso2015}
L.~Ehlers and J.~Mass\'o.
\newblock Matching markets under (in)complete information.
\newblock \emph{Journal of Economic Theory}, 157:\penalty0 295--314, 2015.

\bibitem[Gale and Shapley(1962)]{GaSh62a}
D.~Gale and L.~S. Shapley.
\newblock College admissions and the stability of marriage.
\newblock \emph{The American Mathematical Monthly}, 69\penalty0 (1):\penalty0
  9--15, 1962.

\bibitem[Gusfield and Irving(1989{\natexlab{a}})]{GI89}
D.~Gusfield and R.~Irving.
\newblock \emph{The Stable Marriage Problem: Structure and Algorithms}.
\newblock MIT Press, 1989{\natexlab{a}}.

\bibitem[Gusfield and Irving(1989{\natexlab{b}})]{GuIr89a}
D.~Gusfield and R.~W. Irving.
\newblock \emph{The stable marriage problem: {S}tructure and algorithms}.
\newblock MIT Press, Cambridge, MA, USA, 1989{\natexlab{b}}.

\bibitem[Halld{\'{o}}rsson et~al.(2007)Halld{\'{o}}rsson, Iwama, Miyazaki, and
  Yanagisawa]{HIY07a}
M.~M. Halld{\'{o}}rsson, K.~Iwama, S.~Miyazaki, and H.~Yanagisawa.
\newblock Improved approximation results for the stable marriage problem.
\newblock \emph{{ACM} Trans. Algorithms}, 3\penalty0 (3), 2007.

\bibitem[Hazon et~al.(2012)Hazon, Aumann, Kraus, and Wooldridge]{HAK+12a}
N.~Hazon, Y.~Aumann, S.~Kraus, and M.~Wooldridge.
\newblock On the evaluation of election outcomes under uncertainty.
\newblock \emph{Artificial Intelligence}, 189:\penalty0 1--18, 2012.

\bibitem[Irving(1994)]{Irv94}
R.~Irving.
\newblock Stable marriage and indifference.
\newblock \emph{Discrete Applied Mathematics}, 48:\penalty0 261--272, 1994.

\bibitem[Li and Conitzer(2015)]{LC15}
Y.~Li and V.~Conitzer.
\newblock Cooperative game solution concepts that maximize stability under
  noise.
\newblock In \emph{AAAI}, pages 979--985, 2015.

\bibitem[Manlove(2013)]{Manl13a}
D.~Manlove.
\newblock \emph{Algorithmics of Matching Under Preferences}.
\newblock World Scientific Publishing Company, 2013.

\bibitem[Manlove(1999)]{manlove99}
D.~F. Manlove.
\newblock Stable marriage with ties and unacceptable partners.
\newblock Technical Report TR-1999-29, University of Glasgow, Department of
  Computing Science, 1999.

\bibitem[Rastegari et~al.(2014)Rastegari, Condon, Immorlica, Irving, and
  Leyton-Brown]{RCIIL-ec14}
B.~Rastegari, A.~Condon, N.~Immorlica, R.~Irving, and K.~Leyton-Brown.
\newblock Reasoning about optimal stable matching under partial information.
\newblock In \emph{Proceedings of the {ACM} Conference on Electronic Commerce
  (EC)}, pages 431--448. ACM, 2014.

\bibitem[Roth and Sotomayor(1990)]{RoSo90a}
A.~E. Roth and M.~A.~O. Sotomayor.
\newblock \emph{Two-Sided Matching: {A} Study in Game Theoretic Modelling and
  Analysis}.
\newblock Cambridge University Press, 1990.

\bibitem[Vizing(1964)]{Vizing64}
V.~G. Vizing.
\newblock On an estimate of the chromatic class of a p-graph.
\newblock \emph{Diskret. Analiz}, 3\penalty0 (7):\penalty0 25--30, 1964.

\bibitem[Welsh and Merino(2000)]{WeMi00a}
D.~J.~A. Welsh and C.~Merino.
\newblock The {P}otts model and the {T}utte polynomial.
\newblock \emph{Journal of Mathematical Physics}, 41\penalty0 (3):\penalty0
  1127--1152, 2000.

\end{thebibliography}
%\bibliographystyle{plainnat}
%
%\bibliography{abb,brandt,group,brill}
% \bibliography{pamas}
%,newbibfile}

\newpage

\end{document}